\documentclass[onecolumn,12pt]{IEEEtran} 
\usepackage{color}
\usepackage{amsfonts,color,pslatex}
\usepackage{amssymb,amsmath,latexsym}

\usepackage[para]{threeparttable}

\newcommand{\tr}{{\rm Tr}}
\newcommand{\gf}{{\rm GF}}

\newcommand{\C}{{\mathcal C}}

\newcommand{\SQ}{{\mathrm{SQ}}}

\newtheorem{theorem}{Theorem}[section]
\newtheorem{lemma}[theorem]{Lemma}

\newtheorem{example}[theorem]{Example}
\newtheorem{proposition}[theorem]{Proposition}

\begin{document}

\title{A  Class of Three-Weight  Cyclic Codes\thanks{Z. Zhou's research was supported by
the Natural Science Foundation of China, Proj. No. 61201243. C. Ding's research was supported by
The Hong Kong Research Grants Council, Proj. No. 600812.}}
\author{Zhengchun Zhou\thanks{Z. Zhou is with the School of Mathematics, Southwest Jiaotong University,
Chengdu, 610031, China (email: zzc@home.swjtu.edu.cn).}
and
Cunsheng Ding\thanks{C. Ding is with the Department of Computer Science
                                                  and Engineering, The Hong Kong University of Science and Technology,
                                                  Clear Water Bay, Kowloon, Hong Kong, China (email: cding@ust.hk).}}

\date{\today}
\maketitle

\begin{abstract}
Cyclic codes are a subclass of linear codes and have applications in consumer electronics,
data storage systems, and communication systems as they have efficient encoding and
decoding algorithms. In this paper, a class of three-weight cyclic codes over $\gf(p)$
whose duals have two zeros is presented, where $p$ is an odd prime. The weight distribution
of this class of cyclic codes is settled. Some of the cyclic codes are optimal. The duals of a
subclass of the cyclic  codes are also studied and proved to be optimal.
\end{abstract}


\section{Introduction}\label{sec-intro}

Throughout this paper, let $m$ and $k$ be positive integers such that $s=m/e$ is odd and $s\geq 3$,
where $e=\gcd(m,k)$. Let  $p$ be an odd prime and $q=p^e$. Let $\pi$ be a primitive element
of the finite field $\gf(q^s)$, where $q^s=p^m$.

\vspace{2mm}

An  $[n,\ell,d]$ linear code over $\gf(p)$ is an
$\ell$-dimensional subspace of $\gf(p)^n$ with minimum (Hamming)
distance $d$.
Let $A_i$ denote the number of codewords with Hamming weight $i$
in a code $\mathcal{C}$ of length $n$. The weight enumerator of $\mathcal{C}$  is defined
by
\begin{eqnarray*}
1+A_1x+A_2x^2+\cdots+A_nx^n.
\end{eqnarray*}
The sequence $(1,A_1,A_2,\cdots,A_{n})$ is called the weight distribution
of the code $\mathcal{C}$.

\vspace{2mm}

An $[n,\ell]$  linear code $\C$ over the finite field  $\gf(p)$ is called cyclic if
$(c_0,c_1, \cdots, c_{n-1}) \in \C$ implies $(c_{n-1}, c_0, c_1, \cdots, c_{n-2})
\in \C$.
By identifying the vector $(c_0,c_1, \cdots, c_{n-1}) \in \gf(p)^n$
with
$$
c_0+c_1x+c_2x^2+ \cdots + c_{n-1}x^{n-1} \in \gf(p)[x]/(x^n-1),
$$
any code $\C$ of length $n$ over $\gf(p)$ corresponds to a subset of $\gf(p)[x]/(x^n-1)$.
The linear code $\C$ is cyclic if and only if the corresponding subset in $\gf(p)[x]/(x^n-1)$
is an ideal of the polynomial residue class ring $\gf(p)[x]/(x^n-1)$.
It is well known that every ideal of $\gf(p)[x]/(x^n-1)$ is principal. Let $\C=(g(x))$,
where $g(x)$ is monic and has the least
degree. Then $g(x)$ is called the {generator polynomial} and
$h(x)=(x^n-1)/g(x)$ is referred to as the {parity-check} polynomial of
$\C$. A cyclic code is called irreducible
if its parity-check polynomial is irreducible  over $\gf(p)$. Otherwise, it is called reducible.

\vspace{2mm}

The weight distributions of both irreducible and reducible  cyclic codes have been interesting subjects
of study for many years. For information on the weight distribution of irreducible cyclic codes, the reader
is referred to \cite{MacWilliams}, \cite{van}, \cite{Schmidt}, and the recent survey \cite{Ding12}. Information
on the weight distribution of reducible cyclic codes could be found in \cite{YCD}, \cite{Feng07}, \cite{Luo081},
\cite{Luo082}, \cite{Ding11}, \cite{Ma11}, and \cite{Wang12}.

\vspace{2mm}

Let $h_{0}(x)$, $h_1(x)$, and $h_2(x)$ be the minimal polynomials of $\pi^{-1}$, $(-\pi)^{-1}$, and
$\pi^{-(p^k+1)/2}$ over $\gf(p)$, respectively. It is easy to show that $h_0(x)$, $h_1(x)$, and $h_2(x)$
are polynomials of  degree $m$  and are pairwise distinct. The cyclic code over $\gf(p)$ with length
$p^m-1$ and parity-check polynomial $h_0(x)h_2(x)$ has been extensively studied and is a three-weight
code in the following cases.
\begin{itemize}
\item When $k$  is even and $e=1$, this three-weight cyclic code is due to Trachtenberg \cite{Trachtenberg}.
\item When $k$ is odd, $e=1$, and $p=3$, the cyclic code is related to some planar functions and is proved
          to have only three nonzero weights by Yuan, Carlet, and Ding \cite{CDY05,YCD}.
\item  When $k$ and $e$ are odd and $p$ is any odd prime, Luo and Feng \cite{Luo081} proved that the code
          has only three nonzero weights.
\end{itemize}

The objective of this paper is to study the cyclic code over $\gf(q)$ with length $p^m-1$ and  parity-check
polynomial $h_1(x)h_2(x)$. It will be shown that this cyclic code has only three nonzero weights when $k/e$
is odd, or $k$ is even and $e$ is odd. The weight distribution of the proposed cyclic codes will be determined.
Some of the cyclic codes with parity-check polynomials $h_1(x)h_2(x)$ are optimal. The duals of a subclass of
the cyclic codes are also optimal. The three-weight cyclic codes dealt with in this paper may have applications
in  association schemes \cite{CG84} and secret sharing schemes \cite{CDY05}.

This paper is organized as follows.
Section \ref{sec-quad} introduces necessary results on quadratic forms that will be needed later in this
paper. Section \ref{sec-3wtcodes} defines the class of cyclic codes and determines their weight distributions.
Section \ref{sec-duals3wc} studies the duals of a subclass of the cyclic codes. Section \ref{sec-summ}
concludes this paper and makes some comments on this topic.

\section{Quadratic forms over finite fields}\label{sec-quad}

In this section, we give a brief introduction to the theory of quadratic forms over finite fields which is
needed to calculate the weight distribution of the cyclic codes in the sequel. Quadratic forms  have
been well studied (see the monograph \cite{Niddle} and the references therein), and have applications
in sequence design (\cite{Trachtenberg}, \cite{Klapperodd}),  and coding theory (\cite{Feng07},
\cite{Luo081}, \cite{Luo082}).

Identifying $\gf(q^s)$ with the $s$-dimensional
$\gf(q)$-vector space $\gf(q)^s$, a function $Q$ from $\gf(q^s)$ to $\gf(q)$ can be regarded as an $s$-variable polynomial on $\gf(q)$.
The former is called a quadratic form over $\gf(q)$ if the latter is a homogeneous polynomial of degree two in the form
\begin{eqnarray*}
Q(x_1,x_2,\cdots,x_s)=\sum_{1\leq i\leq j\leq s}a_{ij}x_ix_j
\end{eqnarray*}
where $a_{ij}\in \gf(q)$, and we use a basis $\{\beta_1,\beta_2,\cdots,\beta_{s}\}$ of $\gf(q^s)$ over $\gf(q)$ and identify $x=\sum_{i=1}^sx_i\beta_i$ with
the vector $(x_1,x_2,\cdots,x_{s})\in \gf(q)^s$.
The rank of the
quadratic form $Q(x)$ is defined as the codimension of the $\gf(q)$-vector space
\begin{eqnarray*}
V=\{x\in \gf(q^s)| Q(x+z)-Q(x)-Q(z)=0 \textrm{~for~all~}z\in \gf(q^s)\}.
\end{eqnarray*}
That is $|V|=q^{s-r}$ where $r$ is the rank of $Q(x)$.

For a quadratic form $f(x)$ in $s$ variables over $\gf(q)$, there exists a
symmetric matrix $A$ of order $s$ over $\gf(q)$ such that $f(x)=XAX'$,
where $X=(x_1,\cdots,x_s)\in \gf(q)^s$ and $X'$ denotes the transpose
of $X$. For a symmetric matrix $A$ of order $s$ over $\gf(q)$, it is known  that
 there is a nonsingular matrix  $T$ of order $s$ such that $TAT'$
is a diagonal matrix \cite{Niddle}. Under the nonsingular linear
substitution $X=ZT$ with $Z=(z_1,z_2,\cdots,z_s)\in \gf(q)^s$, we then have
\begin{eqnarray}\label{eqn_f_homogenous}
f(x)=ZTAT'Z'=\sum_{i=1}^rd_iz_i^2
\end{eqnarray}
where $r$ is the rank of $f(x)$  and  $d_i\in \gf(q)^*$. Let $\Delta=d_1d_2\cdots d_r$ for
$r\geq 1$ and $\Delta=1$ for $r=0$. Let $\eta_1$ denote the quadratic
multiplicative character of $\gf(q)$. Then $\eta_1(\Delta)$ is an invariant of $A$  under
the conjugate action of $\mathcal{M}\in {\textrm{GL}}_s(\gf(q))$. The following results are useful in the sequel.

\begin{lemma}(\cite{Niddle}, \cite{Luo082})\label{lemma_quadra_basis}
With the notations as above, we have
\begin{eqnarray*}
\sum_{x\in \gf(q^s)}\zeta_p^{\tr_{q/p}(f(x))}=\left\{\begin{array}{cc}
\eta_1(\Delta)q^{s-{r/2}},&\textrm{~if~} q\equiv 1~(\bmod~4),   \\
\eta_1(\Delta)(\sqrt{-1})^{r}q^{s-{r/2}},&\textrm{~if~} q\equiv 3~(\bmod~4)                                                                                 \\
\end{array}\right.
\end{eqnarray*}
for any quadratic form $f(x)$ in $s$ variables of rank $r$ over $\gf(q)$, where  $\zeta_p$ is a primitive $p$-th  root of unity,
and $\tr_{q/p}(x)$ denotes the trace function from $\gf(q)$ to $\gf(p)$.
\end{lemma}

\begin{lemma}\label{lemma_quadra_important}
Let  $f(x)$ be a quadratic form  in $s$ variables  of rank $r$ over $\gf(q)$.
\begin{itemize}
\item If $r$ is even, then
\begin{eqnarray*}
\sum_{y\in \gf(p)^*}\sum_{x\in \gf(q^s)}\zeta_p^{\tr_{q/p}(y f(x))}=\pm (p-1) q^{s-r/2}.
\end{eqnarray*}
\item If $r$ and $e$ are odd, then
\begin{eqnarray*}
\sum_{y\in \gf(p)^*}\sum_{x\in \gf(q^s)}\zeta_p^{\tr_{q/p}(y f(x))}&&=0.
\end{eqnarray*}
\end{itemize}
\end{lemma}

\begin{proof}
By a nonsingular linear substitution as in (\ref{eqn_f_homogenous}), we have
$f(x)=\sum_{i=1}^{r}d_iz_i^2$, where $d_i\in \gf(q)^*$ and $(z_1,z_2,\cdots,z_r)\in \gf(q)^r$.
Note that, for each $y\in \gf(p)^*$, $yf(x)$ ia a quadratic form over $\gf(q)$ with rank $r$ and $yf(x)=\sum_{i=1}^{r}(yd_i)z_i^2$.
According to Lemma \ref{lemma_quadra_basis}, we have
\begin{eqnarray*}
\sum_{y\in \gf(p)^*}\sum_{x\in \gf(q^s)}\zeta_p^{\tr_{q/p}(y f(x))}=\sum_{x\in \gf(q^s)}\zeta_p^{\tr_{q/p}(f(x))}\sum_{y\in \gf(p)^*}\eta_1(y^r).
\end{eqnarray*}
Thus, when $r$ is even,
\begin{eqnarray*}
\sum_{y\in \gf(p)^*}\sum_{x\in \gf(q^s)}\zeta_p^{\tr_{q/p}(y f(x))}=\pm (p-1) q^{s-r/2}.
\end{eqnarray*}
On the other hand, when $r$ and $e$ are both odd,
\begin{eqnarray*}
\lefteqn{\sum_{y\in \gf(p)^*}\sum_{x\in \gf(q^s)}\zeta_p^{\tr_{q/p}(y f(x))}} \\
&&=\sum_{x\in \gf(q^s)}\zeta_p^{\tr_{q/p}(f(x))}\sum_{y\in \gf(p)^*}\eta_1(y^r)\\
&&=\sum_{x\in \gf(q^s)}\zeta_p^{\tr_{q/p}(f(x))}\sum_{y\in \gf(p)^*}\eta_1(y)\\
&&=\sum_{x\in \gf(q^s)}\zeta_p^{\tr_{q/p}(f(x))}\sum_{y\in \gf(p)^*}\eta_0(y)\\
&&=0
\end{eqnarray*}
where $\eta_0$ is the quadratic multiplicative character of $\gf(p)$ and in the third identity we used the fact that
$\eta_0(x)=\eta_1(x)$ for any $x\in \gf(p)^*$ since $e$ is odd.
\end{proof}


\section{The class of three-weight cyclic codes and their weight distribution}\label{sec-3wtcodes}

We follow the notations fixed in Section \ref{sec-intro}.
From now on, we always assume that $\lambda$ is a fixed nonsquare in $\gf(q)$. Note that $s$ is odd,
thus  $\lambda$ is also a nonsquare in $\gf(q^s)$. Let $\textrm{SQ}$ denote the set of all squares in $\gf(q^s)^*$. Then
$\lambda x$ runs through all nonsquares in $\gf(q^s)$ as $x$ runs through $\textrm{SQ}$.
The following result is easy to prove and is useful in the sequel.

\vspace{2mm}

\begin{proposition}\label{Propo_on_lambda}
$\lambda^{(1+p^{k})/ 2}=\lambda$ if $k/e$ is even, and
$\lambda^{(1+p^{k})/ 2}=-\lambda$ otherwise.
\end{proposition}

\vspace{2mm}

By Delsarte's Theorem \cite{Delsarte}, the code $\C$ with the parity-check polynomial $h_1(x)h_2(x)$ can be expressed as
\begin{eqnarray}\label{eqn_def_C}
\C=\{{\bf c}_{(a,b)}| a,b\in \gf(q^s)\}
\end{eqnarray}
where
\begin{eqnarray*}
{\bf c}_{(a,b)}=\left(\tr_{q^s/p}\left(a(-\pi)^t+b\pi^{{(p^k+1)t/ 2}}\right)\right)_{t=0}^{q^s-2}.
\end{eqnarray*}

In terms of
exponential sums, the weight of the codeword ${\bf c}_{(a,b)}=(c_0,c_1,\cdots,c_{q^s-2})$ in $\C$ is given by
\begin{eqnarray*}
\lefteqn{{\textrm{WT}}({\bf c}_{(a,b)}) } \\
&&=\#\{0\leq t\leq q^s-2: c_t \neq 0\}\\
&&=q^s-1-\frac{1}{p}\sum_{t=0}^{q^s-2}\sum_{y\in {\gf(p)}}\zeta_p^{y c_t}\\
&&=q^s-1-\frac{1}{p}\sum_{y\in {\gf(p)}}\sum_{t=0}^{(q^s-3)/2}\left(\zeta_p^{\tr_{q^s/p}(y a\pi^{2t}+y b (\pi^{2t})^{(p^k+1)/2})}+ \zeta_p^{\tr_{q^s/p}(-y a \pi \pi^{2t}+y b (\pi \pi^{2t})^{(p^k+1)/2})}\right)\\
&&=q^s-1-\frac{1}{p}\sum_{y\in {\gf(p)}}\sum_{x\in {\SQ}}\left(\zeta_p^{\tr_{q^s/p}(y a x+y b x^{(p^k+1)/2})}+ \zeta_p^{\tr_{q^s/p}(-y a \pi x+y b (\pi x)^{(p^k+1)/2})}\right)\\
&&=q^s-1-\frac{1}{p}\sum_{y\in {\gf(p)}}\sum_{x\in {\SQ}}\left(\zeta_p^{\tr_{q^s/p}(y a x+y b x^{(p^k+1)/2})}+ \zeta_p^{\tr_{q^s/p}(-y a \lambda x+y b  (\lambda x)^{(p^k+1)/2})}\right)\\
&&=q^s-1-\frac{1}{2p}\sum_{y\in {\gf(p)}}\sum_{x\in \gf(q^s)^*}\left(\zeta_p^{\tr_{q^s/p}(y a x^2+y b x^{p^{k+1}})}+ \zeta_p^{\tr_{q^s/p}(-y a \lambda x^2+y b \lambda^{(p^k+1)/2} x^{p^k+1})} \right)\\
&&=p^m-p^{m-1}-\frac{1}{2p}\sum_{y\in {\gf(p)^*}}\sum_{x\in \gf(q^s)}\left(\zeta_p^{\tr_{q^s/p}(y ax^2+y b x^{p^k+1})}+ \zeta_p^{\tr_{q^s/p}(-y a \lambda x^2+
y  b \lambda^{(p^k+1)/2} x^{p^k+1})}\right)
\end{eqnarray*}
where  in the fifth identity we used the fact that both $\pi x$ and $\lambda x$
run through all nonsquares in $\gf^*(q^s)$ as $x$ runs through $\textrm{SQ}$. It then follows from Proposition \ref{Propo_on_lambda} that
\begin{itemize}
\item when $k/e$ is even,
\begin{eqnarray}\label{eqn_weight_even_case}
{\textrm{WT}}({\bf c}_{(a,b)})\nonumber=p^m-p^{m-1}-\frac{1}{2p}S(a,b)
\end{eqnarray}
where
\begin{eqnarray}\label{eqn(s)(a,b)}
S(a,b)=\sum_{y\in {\gf(p)^*}}\sum_{x\in \gf(q^s)}\left(\zeta_p^{\tr_{q^s/p}(y ax^2+y b x^{p^k+1})}+ \zeta_p^{\tr_{q^s/p}(-y a \lambda x^2+
y  b \lambda x^{p^k+1})}\right);
\end{eqnarray}

\item when $k/e$ is odd,
\begin{eqnarray}\label{eqn_weight_odd_case}
&&{\textrm{WT}}({\bf c}_{(a,b)})\nonumber=p^m-p^{m-1}-\frac{1}{2p}T(a,b)
\end{eqnarray}
where
\begin{eqnarray}\label{eqn_T(a,b)}
T(a,b)=\sum_{y\in {\gf(p)^*}}\sum_{x\in \gf(q^s)}\left(\zeta_p^{\tr_{q^s/p}(y ax^2+y b x^{p^k+1})}+ \zeta_p^{\tr_{q^s/p}(-y a \lambda x^2-
y  b \lambda x^{p^k+1})}\right).
\end{eqnarray}

\end{itemize}

\vspace{2mm}

Based on the discussions above,  the weight distribution of the code $\C$ is completely determined by
the value distribution of $S(a,b)$ and $T(a,b)$.
To calculate the value distribution of $S(a,b)$ and $T(a,b)$,
we need a series of lemmas.
Before introducing them,  we define
\begin{eqnarray}\label{eqn_Q_{a,b}}
Q_{a,b}(x)=\tr_{q^s/q}(ax^{2}+bx^{1+p^k}), ~~x\in \gf{(q^s)}.
\end{eqnarray}
for each $(a,b)\in \gf(q^s)^2$.

\vspace{2mm}

\begin{lemma}(\cite{Feng07}, \cite{Luo082})\label{lemma_quadra_rank}
For any $(a,b)\in \gf(q^s)^2\setminus \{(0,0)\}$,  the function $Q_{a,b}$ of (\ref{eqn_Q_{a,b}}) is a quadratic form over $\gf(q)$
 with rank  $s, s-1$, or $s-2$.
\end{lemma}

\vspace{2mm}

\begin{lemma}\label{Lemma_value(s)(a,b)}
Let $k$ be even and $e$ be odd, and  let $S(a,b)$ be defined by (\ref{eqn(s)(a,b)}). Then for any
$(a,b)\neq  (0,0)$, $S(a,b)$  takes on only the values from the set $\{0,\pm (p-1)p^{(m+e)/2}\}$.
\end{lemma}

\begin{proof}
According to the definition of $S(a,b)$, we have
\begin{eqnarray*}
S(a,b)=\sum_{y\in {\gf(p)^*}}\sum_{x\in \gf(q^s)}\left(\zeta_p^{\tr_{q/p}(y Q_{a,b}(x))}+ \zeta_p^{\tr_{q/p}(y\lambda Q_{-a,b}(x))}\right)
\end{eqnarray*}
where $Q_{a,b}(x)$ is given by (\ref{eqn_Q_{a,b}}).
We now prove that at least one of the quadratic forms $Q_{a,b}$ and $Q_{-a,b}$ has rank $s$.
When $b=0$, it is easy to check that both $Q_{a,b}$ and $Q_{-a,b}$ have rank $s$ for any
nonzero $a$.
When $b\neq 0$, suppose on the contrary that both $Q_{a,b}$ and $Q_{-a,b}$ have rank less than $m$.
Then there are two nonzero elements $x_1,x_2\in \gf{(q^s)}$
such that
\begin{eqnarray}\label{eqn_Q(x+z)-Q(x)_1}
Q_{a,b}(x_1+z)-Q_{a,b}(x_1)-Q_{a,b}(z)=0,~~ \forall z\in \gf(q^s)
\end{eqnarray}
and
\begin{eqnarray}\label{eqn_Q(x+z)-Q(x)_2}
Q_{-a,b}(x_2+z)-Q_{-a,b}(x_2)-Q_{-a,b}(z)=0,~~ \forall z\in \gf(q^s).
\end{eqnarray}
Note that
\begin{eqnarray*}
Q_{a,b}(x+z)-Q_{a,b}(x)-Q_{a,b}(z)=\tr_{q^s/q}(z(2ax+bx^{p^k}+b^{p^{-k}}x^{p^{-k}})).
\end{eqnarray*}
It then follows from (\ref{eqn_Q(x+z)-Q(x)_1}) and (\ref{eqn_Q(x+z)-Q(x)_2}), respectively, that
\begin{eqnarray*}
b^{p^k}x_1^{p^{2k}}+2a^{p^{k}}x_1^{p^k}+bx_1=0
\end{eqnarray*}
and
\begin{eqnarray*}
b^{p^k}x_2^{p^{2k}}-2a^{p^{k}}x_2^{p^k}+bx_2=0.
\end{eqnarray*}
Combining these two equations (the first one times $x_2^{p^k}$ plus the second one
times $x_1^{p^k}$) leads to
\begin{eqnarray}\label{eqn_u}
u (u^{p^{k}-1}+1)=0
\end{eqnarray}
where
\begin{eqnarray*}
u=bx_1x_2(x_1^{p^k-1}+x_2^{p^k-1}).
\end{eqnarray*}
Note that $x^{p^k-1}\neq -1$ for any $x\in \gf(q^s)^*$ since $k/e$ is even and $s$ is odd. It then follows that
 $u\neq 0$ and  $u^{p^{k}-1}+1\neq 0$. This is a contradiction with (\ref{eqn_u}).
Thus at least one of the quadratic forms $Q_{a,b}$ and $Q_{-a,b}$ has rank $s$
for any $(a,b)\neq (0,0)$. On the other hand, by Lemma \ref{lemma_quadra_important}, $S(a,b)\neq 0$ only if
$Q_{a,b}$ or $Q_{-a,b}$ has even rank. Thus, $S(a,b)=\pm (p-1)p^{(m+e)/2}$
if $Q_{a,b}$ has rank $s$ and $Q_{-a,b}$ has rank $s-1$ or $Q_{a,b}$ has rank $s-1$ and $Q_{-a,b}$ has rank $s$,
and otherwise $S(a,b)=0$. This completes the proof.
\end{proof}

\vspace{2mm}

\begin{theorem}\label{theorem_distribution(s)(a,b)}
Let $k$ be even and $e$ be odd. Then the value distribution
of $S(a,b)$ in (\ref{eqn(s)(a,b)}) is given by
\begin{eqnarray*}
\begin{array}{cccl}
2(p-1)p^{m}&\textrm{occurring}&~~1~~&\textrm{time}\\
(p-1)p^{(m+e)/ 2}&\textrm{occurring}&~~(p^{m-e}+p^{(m-e)/ 2})(p^m-1)~~&\textrm{times}\\
-(p-1)p^{(m+e)/ 2}&\textrm{occurring}&~~(p^{m-e}-p^{(m-e)/ 2})(p^m-1)~~&\textrm{times}\\
0&\textrm{occurring}&~~(p^m-2p^{m-e}+1)(p^m-1)~~&\textrm{times}.
\end{array}
\end{eqnarray*}
\end{theorem}

\begin{proof}
It is clear that $S(0,0)=2(p-1)p^{m}$. According to Lemma \ref{Lemma_value(s)(a,b)}, we define
\begin{eqnarray*}
N_\epsilon=\#\{(a,b)\in \gf(q^s)^2\setminus (0,0)| ~~S(a,b)=\epsilon(p-1)p^{(m+e)/2}\}
\end{eqnarray*}
where $\epsilon=\pm 1$. Then we have
\begin{eqnarray}\label{eqn_final(s)(a,b)1}
\sum_{a,b}S(a,b)=2(p-1)p^{m}+(N_{1}-N_{-1})(p-1)p^{(m+e)/2}
\end{eqnarray}
and
\begin{eqnarray}\label{eqn_final(s)(a,b)2}
\sum_{a,b}S^2(a,b)=4(p-1)^2p^{2m}+(N_{1}+N_{-1})(p-1)^2p^{m+e}.
\end{eqnarray}
On the other hand, it follows from (\ref{eqn(s)(a,b)}) that
\begin{eqnarray}\label{eqn_final(s)(a,b)3}
\sum_{a,b}S(a,b)=2(p-1)p^{2m}
\end{eqnarray}
and
\begin{eqnarray}\label{eqn_final(s)(a,b)4}
\sum_{a,b}S^2(a,b)=p^{2m}(\#S_1+\#S_2+\#S_3+\#S_4)
\end{eqnarray}
where
\begin{eqnarray*}
&&S_1=\{(y_1,y_2,x_1,x_2)\in \Gamma| y_1x_1^2-y_2x_2^2=0, ~~y_1x_1^{p^k+1}-y_2x_2^{p^k+1}=0\}, \\
&&S_2=\{(y_1,y_2,x_1,x_2)\in \Gamma| y_1x_1^2+\lambda y_2x_2^2=0, ~~y_1x_1^{p^k+1}-\lambda y_2x_2^{p^k+1}=0\},\\
&&S_3=\{(y_1,y_2,x_1,x_2)\in \Gamma| -\lambda y_1x_1^2- y_2x_2^2=0, ~~\lambda y_1x_1^{p^k+1}- y_2x_2^{p^k+1}=0\},\\
&&S_4=\{(y_1,y_2,x_1,x_2)\in \Gamma| -\lambda y_1x_1^2+\lambda y_2x_2^2=0, ~~\lambda y_1x_1^{p^k+1}-\lambda y_2x_2^{p^k+1}=0\}.
\end{eqnarray*}
Herein, $\Gamma=\gf(p)^*\times \gf(p)^* \times \gf(q^s) \times \gf(q^s)$. It is not hard to prove that
\begin{eqnarray}\label{eqn(s)_1S4}
\#S_1=\#S_4=(p-1)^2p^m
\end{eqnarray}
and
\begin{eqnarray}\label{eqn(s)_2S3}
\#S_2=\#S_3=(p-1)^2.
\end{eqnarray}
Combining  Equations (\ref{eqn_final(s)(a,b)1})--(\ref{eqn(s)_2S3}), we get
\begin{eqnarray*}
&&N_1=(p^{m-e}+p^{(m-e)/ 2})(p^m-1),\\
&&N_{-1}=(p^{m-e}-p^{(m-e)/ 2})(p^m-1).
\end{eqnarray*}
Summarizing the discussion above completes the proof of this theorem.
\end{proof}

\vspace{3mm}

\begin{lemma}\label{Lemma_value_T(a,b)}
Let $k/e$ be odd and $T(a,b)$ be defined by (\ref{eqn_T(a,b)}). Then for any $(a,b)\in \gf(q^s)^2\setminus \{(0,0)\}$,  $T(a,b)$
takes on only the values from the set $\{0,\pm 2(p-1)p^{(m+e)/2}\}$.
\end{lemma}

\begin{proof}
According to the definition of $T(a,b)$, we have
\begin{eqnarray*}
T(a,b)=\sum_{y\in {\gf(p)^*}}\sum_{x\in \gf(q^s)}\left(\zeta_p^{\tr_{q/p}(y Q_{a,b}(x))}+ \zeta_p^{\tr_{q/p}(- y\lambda Q_{a,b}(x))}\right)
\end{eqnarray*}
where $Q_{a,b}(x)$ is given by (\ref{eqn_Q_{a,b}}). By Lemma \ref{lemma_quadra_rank}, for any $(a,b)\neq (0,0)$, the possible rank of $Q_{a,b}(x)$
is $s$, $s-1$, or $s-2$. Note that, for any $y\in \gf(p)^*$, the quadratic forms $y Q_{a,b}(x)$ and $-\lambda y Q_{a,b}(x)$
have the same rank with $Q_{a,b}(x)$. When $Q_{a,b}(x)$ has even rank $s-1$, by Lemma \ref{lemma_quadra_important}, we have
$T(a,b)=\pm 2(p-1)p^{(m+e)/2}$. When $Q_{a,b}(x)$ has odd rank $s$ or $s-2$, we distinguish between the following
two cases to show that $T(a,b)=0$. Case 1:  $e$ is odd. In this case, it follows again from
Lemma \ref{lemma_quadra_important}  that $T(a,b)=0$. Case 2: $e$ is even. In this case,
$-1$ is a square in $\gf(q)$ and thus $-\lambda$ is a nonsquare in $\gf(q)$. It then follows from  Lemma \ref{lemma_quadra_basis} that
\begin{eqnarray*}
&&\sum_{x\in \gf(q^s)}\left(\zeta_p^{\tr_{q/p}(y Q_{a,b}(x))}+ \zeta_p^{\tr_{q/p}(- y\lambda Q_{a,b}(x))}\right)\\
&&~~~=(\eta_1(y)+\eta_1(-y\lambda))\sum_{x\in \gf(q^s)}\zeta_p^{\tr_{q/p}(Q_{a,b}(x))}\\
&&~~~=0
\end{eqnarray*}
for any $y\in \gf(p)^*$. Thus  $T(a,b)=0$. This completes the proof.
\end{proof}
\vspace{3mm}

\begin{theorem}\label{Theorem_distribution_T(a,b)}
Let $k/e$ be odd. Then the value distribution
of $T(a,b)$ in  (\ref{eqn_T(a,b)}) is given by
\begin{eqnarray*}
\begin{array}{cccl}
2(p-1)p^{m}&\textrm{occurring}&~~1~~&\textrm{time}\\
2(p-1)p^{(m+e)/ 2}&\textrm{occurring}&~~\frac{1}{2}(p^{m-e}+p^{(m-e)/ 2})(p^m-1)~~&\textrm{times}\\
-2(p-1)p^{(m+e)/ 2}&\textrm{occurring}&~~\frac{1}{2}(p^{m-e}-p^{(m-e)/ 2})(p^m-1)~~&\textrm{times}\\
0&\textrm{occurring}&~~(p^m-p^{m-e}+1)(p^m-1)~~&\textrm{times}.
\end{array}
\end{eqnarray*}
\end{theorem}

\begin{proof}
It is clear that $T(0,0)=2(p-1)p^{m}$. According to Lemma \ref{Lemma_value_T(a,b)}, we define
\begin{eqnarray*}
n_\epsilon=\#\{(a,b)\in \gf(q^s)^2\setminus (0,0)| ~~T(a,b)= 2\epsilon(p-1)p^{(m+e)/2}\}
\end{eqnarray*}
where $\epsilon=\pm 1$. Then we have
\begin{eqnarray}\label{eqn_final_T(a,b)1}
\sum_{a,b}T(a,b)=2(p-1)p^{m}+2(n_{1}-n_{-1})(p-1)p^{(m+e)/2}
\end{eqnarray}
and
\begin{eqnarray}\label{eqn_final_T(a,b)2}
\sum_{a,b}T^2(a,b)=4(p-1)^2p^{2m}+4(n_{1}+n_{-1})(p-1)^2p^{m+e}.
\end{eqnarray}
On the other hand, it follows from (\ref{eqn_T(a,b)}) that
\begin{eqnarray}\label{eqn_final_T(a,b)3}
\sum_{a,b}T(a,b)=2(p-1)p^{2m}
\end{eqnarray}
and
\begin{eqnarray}\label{eqn_final_T(a,b)4}
\sum_{a,b}T^2(a,b)=p^{2m}(\#T_1+\#T_2+\#T_3+\#T_4)
\end{eqnarray}
where
\begin{eqnarray*}
&&T_1=\{(y_1,y_2,x_1,x_2)\in \Gamma| y_1x_1^2-y_2x_2^2=0, ~~y_1x_1^{p^k+1}-y_2x_2^{p^k+1}=0\},\\
&&T_2=\{(y_1,y_2,x_1,x_2)\in \Gamma| y_1x_1^2+\lambda y_2x_2^2=0, ~~y_1x_1^{p^k+1}+\lambda y_2x_2^{p^k+1}=0\},\\
&&T_3=\{(y_1,y_2,x_1,x_2)\in \Gamma| -\lambda y_1x_1^2- y_2x_2^2=0, ~~-\lambda y_1x_1^{p^k+1}- y_2x_2^{p^k+1}=0\},\\
&&T_4=\{(y_1,y_2,x_1,x_2)\in \Gamma| -\lambda y_1x_1^2+\lambda y_2x_2^2=0, ~~-\lambda y_1x_1^{p^k+1}+\lambda y_2x_2^{p^k+1}=0\}.
\end{eqnarray*}
Herein, $\Gamma=\gf(p)^*\times \gf(p)^* \times \gf(q^s) \times \gf(q^s)$. It is not hard to show that
\begin{eqnarray}\label{eqn_T_2T3}
\#T_1=\#T_2=\#T_3=\#T_4=(p-1)^2p^m.
\end{eqnarray}
Combining  Equations (\ref{eqn_final_T(a,b)1})--(\ref{eqn_T_2T3}), we get
\begin{eqnarray*}
&&n_1=\frac{1}{2}(p^{m-e}+p^{(m-e)/ 2})(p^m-1),\\
&&n_{-1}=\frac{1}{2}(p^{m-e}-p^{(m-e)/ 2})(p^m-1).
\end{eqnarray*}
The value distribution of $T(a,b)$ follows from the discussion above.
\end{proof}

\vspace{3mm}

\begin{table}[!t]
\renewcommand{\arraystretch}{1.5}
\centering
\begin{threeparttable}
\caption{Weight Distribution of $\C$ for even $k$ and odd $e$}\label{Table_even21}
\begin{tabular}{|l|l|}
\hline
 Hamming Weight& Frequency\\
\hline
\hline
  $0$ & $1$ \\
  \hline
 $p^m-p^{m-1}-\frac{p-1}{2} p^{(m+e-2)/2 }$ & $(p^{m-e}+p^{(m-e)/ 2})(p^m-1)$\\
\hline
$p^m-p^{m-1}$ & $(p^m-2p^{m-e}+1)(p^m-1)$\\
\hline
$p^m-p^{m-1}+\frac{p-1}{2} p^{(m+e-2)/ 2}$ & $(p^{m-e}-p^{(m-e)/ 2})(p^m-1)$\\
\hline
\end{tabular}
\begin{tablenotes}
\end{tablenotes}
\end{threeparttable}
\end{table}

\begin{theorem}\label{Theorem_main}
Let $k$ be even and $e$ be odd. Then the code $\C$ is a three-weight $p$-ary cyclic code with parameters
$[p^m-1,2m,p^m-p^{m-1}-\frac{p-1}{2} p^{(m+e-2)/ 2}]$. Moreover the weight distribution of $\C$ is given
in Table \ref{Table_even21}.
\end{theorem}

\begin{proof}
The length and dimension of $\C$ follow directly from its definition. The minimal distance and weight distribution of $\C$
follow from Equation (\ref{eqn_weight_even_case}) and Theorem \ref{theorem_distribution(s)(a,b)}.
\end{proof}

\vspace{3mm}

The following are some examples of the codes.

\begin{example}
Let $p=3$ and $m=3$, and $k=2$. Then the code $\C$ is a $[26,6,15]$ code over $\gf(3)$ with the weight enumerator
\begin{eqnarray*}
&&1+ 312^{15}+260 x^{18}+156 x^{21}.
\end{eqnarray*}
It has the same parameters with the best known cyclic codes
in the Database of best linear codes known maintained by
         Markus Grassl at http://www.codetables.de/.      It is also optimal since the upper bound is $15$.
\end{example}

\begin{example}
Let $p=3$, $m=5$ and $k=4$. Then the code $\C$ is a $[242,6,153]$ code over $\gf(3)$ with the weight enumerator
\begin{eqnarray*}
&&1+ 21780^{153}+19844 x^{162}+17424 x^{171}.
\end{eqnarray*}
It has the same parameters with the best known cyclic codes
in the Database. It is optimal or almost optimal since the upper bound on the minimal distance of  any ternary linear code with
length $242$ and dimension $6$ is $154$.
\end{example}

\begin{example}
Let $p=5$, $m=3$, and $k=2$. Then the code $\C$ is a $[124,6,90]$ code over $\gf(5)$ with the weight enumerator
\begin{eqnarray*}
&&1+ 3720^{90}+9424 x^{100}+2480 x^{110}.
\end{eqnarray*}
The best known linear code over $\gf(5)$ with length $124$ and dimension $6$ has minimal distance 95.
\end{example}

\vspace{3mm}

\begin{table}[!t]
\renewcommand{\arraystretch}{1.5}
\centering
\begin{threeparttable}
\caption{Weight Distribution of $\C$ for odd $k/e$}\label{Table_even22}
\begin{tabular}{|l|l|}
\hline
 Hamming Weight& Frequency\\
\hline
\hline
  $0$ & $1$ \\
  \hline
 $p^m-p^{m-1}-{(p-1)}p^{(m+e-2)/ 2}$ & $\frac{1}{2} (p^{m-e}+p^{(m-e)/ 2})(p^m-1)$\\
\hline
$p^m-p^{m-1}$ & $(p^m-p^{m-e}+1)(p^m-1)$\\
\hline
$p^m-p^{m-1}+{(p-1)}p^{(m+e-2)/ 2}$ & $\frac{1}{2} (p^{m-e}-p^{(m-e)/ 2})(p^m-1)$\\
\hline
\end{tabular}
\begin{tablenotes}
\end{tablenotes}
\end{threeparttable}
\end{table}

\begin{theorem}\label{Theorem_main2}
Let $k/e$ be odd. Then the code $\C$ is a three-weight $p$-ary cyclic code with parameters
$[p^m-1, 2m, p^m-p^{m-1}-{(p-1)}p^{(m+e-2)/2}]$. Moreover the weight distribution of
$\C$ is given in Table \ref{Table_even22}.
\end{theorem}

\begin{proof}
The length and dimension of $\C$ follow directly from its definition. The minimal distance and weight distribution of $\C$
follow from Equation (\ref{eqn_weight_odd_case}) and Theorem
 \ref{Theorem_distribution_T(a,b)}.
\end{proof}

\vspace{3mm}

\begin{example}
Let $p=3$ and $m=6$, and $k=2$. Then the code $\C$ is a $[728,12,432]$ code over $\gf(3)$ with the weight enumerator
\begin{eqnarray*}
&&1+ 32760^{432}+472472 x^{486}+26208 x^{540}.
\end{eqnarray*}
\end{example}

\begin{example}
Let $p=5$ and $m=3$, and $k=1$. Then the code $\C$ is a $[124,6,80]$ code over $\gf(5)$ with the weight enumerator
\begin{eqnarray*}
&&1+ 1860^{80}+12524 x^{100}+1240 x^{120}.
\end{eqnarray*}
\end{example}

\begin{example}
Let $p=3$ and $m=9$, and $k=3$. Then the code $\C$ is a $[19682,18,12879]$ code over $\gf(3)$ with the weight enumerator
\begin{eqnarray*}
&&1+ 7439796^{12636}+373072310 x^{13122}+6908382 x^{13608}.
\end{eqnarray*}
\end{example}

\section{The duals of a subclass of the cyclic codes}\label{sec-duals3wc}

In this section, we study the duals of a subclass of the cyclic codes presented in this paper
and prove that they are optimal ternary linear codes.

\begin{theorem}
Let $p=3$, $k$ be even and $e$ be odd. Then the dual $\C^\perp$ of the cyclic code $\C$ in (\ref{eqn_def_C})
is an optimal ternary code with parameters $[3^m-1,3^m-1-2m,4]$.
\end{theorem}

\begin{proof}
We only need to prove that  $\C^\perp$ has minimal distance $4$.
Clearly, the minimal distance $d$ of the dual of $\C^\perp$ cannot be $1$.
Let $u={(p^{k}+1)/2}$. Then $\gcd(u,p^m-1)=1$ since $k$ is even and $m$ is odd.
By the definition of $\C$, the
code $\C^\perp$ has a codeword of Hamming weight $2$ if and only if there exist two
elements $c_1,c_2\in \gf(3)^*$ and two distinct integers $0\leq t_1<t_2\leq p^m-2$ such that
\begin{eqnarray}\label{eqn_system_1}
\left\{
\begin{array}{llll}
c_1(-\pi)^{t_1}&+&c_2(-\pi)^{t_2}&=0\\
c_1\pi^{ut_1}&+&c_2\pi^{ut_2}&=0.
\end{array} \right.\ \
\end{eqnarray}
Note that $\gcd(u,p^m-1)=1$ and $t_1\neq t_2$. It follows from the second equation of
(\ref{eqn_system_1}) that $c_1=c_2$ and  $t_2=t_1+ {(p^m-1)/ 2}$. Then the first equation becomes
$2c_1(-\pi)^{t_1}=0$, which is impossible. Thus  the
code $\C^\perp$ does not have a codeword of Hamming weight $2$.

We now prove that $\C^\perp$ has no codeword of weight $3$. Otherwise,
there exist three elements $c_1,c_2,c_3$ in $\gf(3)^*$ and three distinct integers
$0\leq t_1<t_2<t_3\leq p^m-2$ such that
 \begin{eqnarray}\label{eqn_system_2}
\left\{
\begin{array}{llll}
c_1(-\pi)^{t_1}&+c_2(-\pi)^{t_2}&+c_3(-\pi)^{t_3} &=0\\
c_1\pi^{ut_1}&+c_2\pi^{ut_2}&+c_3\pi^{ut_3} &=0.
\end{array} \right.\ \
\end{eqnarray}
Due to symmetry it is sufficient to consider the following two cases.

{\textit{Case A}}, when $c_1=c_2=c_3=1$: In this case, (\ref{eqn_system_2}) becomes
 \begin{eqnarray}\label{eqn_system_3}
\left\{
\begin{array}{llllll}
(-\pi)^{t_1}&+(-\pi)^{t_2}&+(-\pi)^{t_3} &=0\\
\pi^{ut_1}&+\pi^{ut_2}&+\pi^{ut_3} &=0.
\end{array} \right.\ \
\end{eqnarray}
Let $x_i=\pi^{t_i}$ for $i=1,2,3$. Then $x_1,x_2,x_3\in \gf(3^m)^*$ and  are pairwise distinct.
Without loss of generality, we only need to consider the following two subcases.
\begin{itemize}

\item[1)] $t_1$ is even and $t_2,t_3$ are odd. In this subcase, we have
 \begin{eqnarray}\label{eqn_system_32}
\left\{
\begin{array}{ll}
x_1-x_2-x_3&=0\\
x_1^u+x_2^u+x_3^u &=0
\end{array} \right.\ \
\end{eqnarray}
which yields
\begin{eqnarray*}
(x_2+x_3)^u=-x_2^u-x_3^u.
\end{eqnarray*}
Thus
\begin{eqnarray*}
(x_2+x_3)^{2u}=(-x_2^u-x_3^u)^2
\end{eqnarray*}
which leads to
\begin{eqnarray}\label{eqn_case_1}
x_2x_3\left(x_2^{(p^k-1)/2}-x_3^{(p^k-1)/2}\right)^2=0.
\end{eqnarray}
It then follows that $(x_3/x_2)^{p^k-1}=1$. Note that $\gcd(m,k)=1$,
and $x_2,x_3$ are both nonsquare in $\gf(3^m)$ since $t_2,t_3$ are odd.
Thus $x_2=x_3$. This is a contradiction to the fact that $x_2 \neq x_3$.

\item[2)] $t_1$, $t_2$, and  $t_3$ are even. Similarly, in this subcase, we can arrive at
(\ref{eqn_case_1}) in which $x_2, x_3$ are both squares in $\gf(3^m)$ since $t_2,t_3$ are even.
It then follows from $(x_3/x_2)^{p^k-1}=1$  that $x_2=x_3$. This is again a contradiction.

\end{itemize}

{\textit{Case B}}, when $c_1=c_2=1$ and $c_3=-1$. The proof of this case is similar to Case A. We omit the details here.

Finally, by the Sphere Packing bound, the minimal distance $d\leq 4$. Hence $d=4$. This completes the proof.
\end{proof}

This subclass of ternary cyclic codes $\C^\perp$ are optimal in the sense that the minimum distance
is maximal for any ternary linear code with length $3^m-1$ and dimension $3^m-1-2m$.

\begin{example}
Let $p=3$ and $m=3$, and $k=2$. Then the code $\C^\perp$ is an optimal ternary cyclic
code with parameters
$[26,20,4]$ and generator polynomial $x^6 + 2x^5 + 2x^3 + x + 2$.
\end{example}

\begin{example}
Let $p=3$ and $m=5$, and $k=4$. Then the code $\C^\perp$ is an optimal ternary cyclic
code with parameters
$[242,10,4]$ and generator polynomial $x^{10} + 2x^9 + x^8 + x^7 + x^6 + 2x^2 + 2$.
\end{example}

\begin{example}
Let $p=3$ and $m=7$,  and $k=6$. Then the code $\C^\perp$ is an optimal ternary cyclic
code with parameters
$[2186,14,4]$ and generator polynomial $x^{14} + 2x^{13} + x^{11} + 2x^{10} + x^6 + x^5 + 2$.
\end{example}

\vspace{3mm}

\section{Summary and concluding remarks}\label{sec-summ}

In this paper, we presented a class of three-weight cyclic codes and determined their weight
distributions. Some of the codes are optimal, and the duals of a subclass of the cyclic codes
are also optimal.

While a lot of two-weight codes were discovered (see \cite{CK86,Ding11,Ding12,LZ12,Kant,Schmidt}),
only a small number of three-weight codes are known (\cite{CG84,Ding11,Ding12,Feng12,McGu}).
It would be good if more three-weight codes are constructed, in view of their applications in association
schemes and secret sharing schemes.


\begin{thebibliography}{99}

\bibitem{CK86} A. R. Calderbank and W. M. Kantor, ``The geometry of two-weight codes,"
{\em Bull. London Math. Soc.}, vol. 18, pp. 97-122, 1986.

\bibitem{CG84} A. R. Calderbank and J. M. Goethals, ``Three-weight codes and association schemes,''
{\em Philips J. Res.,} vol. 39, pp. 143--152,  1984.

\bibitem{CDY05} C. Carlet, C. Ding, and J. Yuan, ``Linear codes from highly nonlinear functions and their
secret sharing schemes,'' {\em IEEE Trans. Inform. Theory,} vol. 51, no. 6, pp. 2089--2102, 2005.

\bibitem{Delsarte} P. Delsarte, ``On subfield subcodes of modified Reed-Solomon codes,''
{\em IEEE Trans. Inform. Theory,} vol. IT-21, no. 5, pp. 575--576, Sep. 1975.

\bibitem{Ding11} C. Ding, Y. Liu, C. Ma, and L. Zeng, ``The weight distributions of the duals of cyclic codes with
two zeros,'' {\em IEEE Trans. Inform. Theory,} vol. 57, No. 12, pp. 8000--8006, Dec.
2011.

\bibitem{Ding12} C. Ding and J. Yang, ``Hamming weights in irreducible cyclic codes,''
\emph {Discrete Mathematics,} vol. 313, no. 4, pp. 434--446, 2013.

\bibitem{Feng07} K. Feng and J. Luo,  ``Weight distribution of some reducible cyclic codes,''
\emph{Finite Fields Appl.,} vol. no. 2, pp. 390--409, Apr. 2008.

\bibitem{Feng12} T. Feng, ``On cyclic codes of length $2^{2^r}-1$ with two zeros whose dual codes have three weights,"
{\em Des. Codes Cryptogr.}, vol. 62,  pp. 253--258, 2012.

\bibitem{McGu} G. McGuire, ``On three weights in cyclic codes with two zeros,"
\emph{Finite Fields Appl.,}  vol. 10, no. 1, pp. 97--104,  Jan. 2004.

\bibitem{Kant}  W. M. Kantor, ``Exponential numbers of two-weight codes, difference sets and symmetric
designs," {\em Discrete Mathematics,} vol. 46, pp. 95--98, 1983.

\bibitem{Klapperodd} A. Klapper, ``Cross-correlations of quadratic form sequences in odd characteristic,''
{\em  Des. Codes Cryptogr.,}, vol. 3, no. 4, pp. 289--305, 1997.

\bibitem{LW12} Z. Liu and X.-W. Wu, ``On a class of three-weight codes with
cryptographic applications," In: {\em Proc. of the 2012 IEEE International Symposium on Information Theory,}
pp. 2551--2555, 2012.

\bibitem{LZ12}  Z. Liu and X. Zeng, ``On a kind of two-weight code,"
{\em European Journal of Combinatorics,} vol. 33, no. 6, pp. 1265--1272, Aug. 2012.

\bibitem{Luo081} J. Luo and  K. Feng, ``Cyclic codes and sequences from generalized
Coulter-Matthews function,'' \emph{ IEEE Trans. Inform. Theory } vol. 54, no. 12, pp. 5345--5353.  Dec. 2008.

\bibitem{Luo082} J. Luo and K. Feng, ``On the weight distributions of two classes of cyclic codes,''
 \emph{IEEE Trans. Inform. Theory,} vol. 54, no. 12, pp. 5332--5344,  Dec. 2008.

\bibitem{Ma11} C. Ma, L. Zeng, Y. Liu, D. Feng, and C. Ding,
``The weight enumerator of a class of cyclic codes,''
\emph{IEEE Trans. Inform. Theory,} vol. 57, no. 1, pp. 397--402, Jan. 2011.

\bibitem{Niddle} R. Lidl and H. Niederreiter, {\em Finite Fields,} Encyclopedia of Mathematics, Vol. 20, Cambridge University Press, Cambridge, 1983.

\bibitem{MacWilliams} F. MacWilliams and J. Seery, ``The weight distributions of some minimal
cyclic codes,'' \emph{IEEE Trans. Inform. Theory,} vol. 27  no. 6, pp. 796--806, 1981.

\bibitem{Schmidt}  B. Schmidt and C. White, ``All two-weight irreducible cyclic codes,'' \emph{Finite
Fields Appl.,} vol. 8, pp. 1--17, 2002.

\bibitem{Trachtenberg} H. M. Trachtenberg, {\em On the crosscorrelation functions of maximal
linear recurring sequences,} Ph.D. dissertation, Univ. South. Calif., Los Angels, 1970.

\bibitem{YCD} J. Yuan, C. Carlet, and C. Ding, ``The weight distribution of a class of linear codes from perfect nonlinear functions,''
\emph{IEEE Trans. Inform. Theory,} {vol. 52}, no. 2, pp. 712--717, Feb. 2006.

\bibitem{Wang12} B. Wang, C. Tang, Y. Qi, Y. X. Yang , and M. Xu,  ``The weight distributions of cyclic codes and elliptic curves,''
\emph{IEEE Trans. Inform. Theory,} {vol. 58}, no. 12, pp. 7253--7259, Dec. 2012.

\bibitem{van} M. van der Vlugt, ``On the weight hierarchy of irreducible cyclic codes,''
\emph{J. Comb. Theory Ser. A,}  vol. 71, no. 1, pp. 159--167, July 1995.


\end{thebibliography}
\end{document}